\documentclass[%
 reprint,
 amsmath,amssymb,
 aps,
]{revtex4-2}
\usepackage{amsthm}
\usepackage{mathrsfs}
\usepackage{subcaption}
\usepackage{placeins}

\usepackage{graphicx}
\usepackage{dcolumn}
\usepackage{bm}

\newtheorem{theorem}{Theorem}

\newtheorem{corollary}[theorem]{Corollary}

\newtheorem{lemma}[theorem]{Lemma}

\newcommand{\lc}[1]{\tilde{#1}}

\usepackage{hyperref}

\begin{document}

\title{Observables for cyclic causal set cosmologies}

\author{Fay Dowker${}^{a,b}$}
\author{Stav Zalel${}^{a}$}%
\thanks{Corresponding author: stav.zalel11@imperial.ac.uk}
\affiliation{%
${}^{a}$Blackett Laboratory, Imperial College London, SW7 2AZ, U.K.\\
${}^{b}$Perimeter Institute, 31 Caroline Street North, Waterloo ON, N2L 2Y5, Canada.}%

\date{\today}

\begin{abstract}
In causal set theory, cycles of cosmic expansion and collapse are modelled by causal sets with  ``breaks'' and ``posts''  and a special role is played by cyclic dynamics in which the universe goes through perpetual cycles. We identify and characterise two algebras of observables for cyclic dynamics in which the causal set universe has infinitely many breaks. The first algebra is constructed from the cylinder sets associated with finite causal sets that  have a single maximal element and offers a new framework for defining cyclic dynamics as random walks on a novel tree.  The second algebra is generated by a collection of stem-sets and offers a physical interpretation of the observables in these models as statements about unlabeled stems with a single maximal element. There are analogous theorems for cyclic dynamics in which the causal set universe has infinitely many posts.
\end{abstract}

\maketitle

\section{Introduction}
Within causal set theory, sequential growth dynamics form a concrete theory space in which to explore ideas for cosmology \cite{Rideout:1999ub,Dowker:2020qqs,Brightwell:2002vw,Dowker:2005gj,Dowker:2005gj,Dowker:2014xga,Ahmed:2009qm,Sorkin:1998hi,Martin:2000js,Dowker:2017zqj,Dowker:2019qiz,Zalel:2020oyf,Bento:2021voo}. 

In this work, we explore a class of sequential growth dynamics that is of particular interest to cosmology: \textit{cyclic} sequential growth  models, in which the causal set (causet) universe goes through perpetual cycles of expansion and contraction punctuated by \textit{breaks}, where a break is an ordered partition $(\lc{A},\lc{B})$ of the causet that  satisfies $a\prec b \ \forall \ a\in \lc{A}, b\in \lc{B}$.  A subclass of causets with infinitely many breaks is that of causets with infinitely many \textit{posts}, where a post is a single causet element with a break immediately below it and a break immediately above it. An illustration is shown in Fig.\ref{postbreak}.

\begin{figure}[h]
    \centering
    \includegraphics[width=0.35\textwidth]{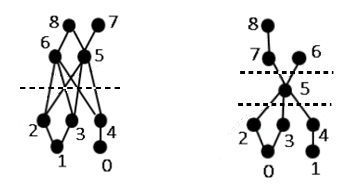}
    \caption{An illustration of breaks and posts. The dotted lines highlight the boundary between the past and the future of the breaks. Left: a causal set with a single break. The subcauset below (above) the dotted line is the past (future) of the break. Right: a causal set with two breaks and a post—the element 5—between them.}%
    \label{postbreak}%
\end{figure}

Cyclic models play a central role in the causal set cosmological paradigm, a heuristic that aims to explain the emergence of a flat, homogeneous and isotropic cosmos directly from the quantum gravity era \cite{Sorkin:1998hi,Martin:2000js,Dowker:2017zqj}. This proposal is one of several in a recent trend to develop theories of cyclic or bouncing cosmologies and determine their implications for fundamental physics. In this work, we take a step towards advancing this school of thought within causal set theory. The cyclic models which we study are those dynamics which give rise to cyclic universes and our aim is to investigate the question: what are the physical observables (covariant events) in these cyclic models?

We note that the cosmological paradigm within which we will be working pertains only to the causal set spacetime, not to any matter living on it. Whether a causal set is enough to give rise to matter degrees of freedom \cite{Rideout:1999ub} or whether one requires additional structure such as a field living on the causal set \cite{X:2017jal,Glaser:2020yfy} is still unknown. Whichever the case may be, this simplified cosmological paradigm will act as a guide to building a causal set cosmology.

In section \ref{sec:background} we review the concepts of sequential growth models, events and covariant events and known results. 
 In section \ref{chap_6_o_subsec} we identify a $\sigma$-algebra $\mathcal{R}_b$ that  forms a complete set of covariant events in cyclic dynamics. We prove that $\mathcal{R}_b$ can be constructed using a particular subcollection of the cylinder sets and discuss the implications of the result for the search for classical and quantum cyclic dynamics. However, we note that the physical interpretation of events in $\mathcal{R}_b$ remains obscure and in section \ref{sec2} we identify a second $\sigma$-algebra of observables, $\mathcal{R}(\widehat{\mathcal{S}})$  that  endows each observable in a cyclic dynamics with a physical interpretation. We prove that this physically meaningful algebra exhausts the 
algebra of covariant events in a well defined way. 
We will consider in detail the general cyclic dynamics defined above in which the infinitely many epochs (cycles) are separated by breaks. At the end of the paper, in section \ref{sec:posts} we describe how our results carry over \textit{mutatis mutandis} to the special case where the epochs are separated by posts.

\section{Background}\label{sec:background}
\subsection{Sequential growth dynamics and covariant events}\label{sec:sequential}

A sequential growth dynamics is a probability space $(\lc{\Omega},\lc{\mathcal{R}},\mu)$. 
$\lc{\Omega}$ is the set of causal sets (causets, for short) on the ground-set $\mathbb{N}$ that satisfy $x\prec y\implies x<y$ for all $x,y\in\mathbb{N}$, where $\prec$ denotes the partial ordering. Let $\lc{C}_n$ denote a causet on ground-set $[0,n-1]$ satisfying $x\prec y\implies x<y$. For each $\lc{C}_n$ there is the cylinder set, 
\begin{align} cyl(\lc{C}_n):=\{\lc{C}\in \lc{\Omega}\, |\, \lc{C}|_{[0,n-1]}=\lc{C}_n\},
\end{align} 
where $\lc{C}|_{[0,n-1]}$ denotes the restriction of $\lc{C}$ to ${[0,n-1]}$. 
$\lc{\mathcal{R}}$ is the $\sigma$-algebra generated by the cylinder sets. 

The measure $\mu$ is the extension--via the Fundamental Theorem of Measure Theory \cite{Kolmogorov:1975}--of the measure on the semi-ring of cylinder sets given by a random walk on \textit{labeled poscau}. Labeled poscau is a directed tree in which each $\lc{C}_n$ is a node and $\lc{D}_m\prec \lc{C}_n \iff \lc{D}_m$ is a stem (a finite down-set) in $\lc{C}_n$ (Fig.\ref{labeled_poscau}). Then  ${\mu}(cyl(\lc{C}_n))=\mathbb{P}(\lc{C}_n)$, where $\mathbb{P}(\lc{C}_n)$ is the probability that the random walk goes through the node $\lc{C}_n$ \cite{Rideout:1999ub,Dowker:2019qiz}.

\begin{figure}[htpb]
\centering
    \includegraphics[width=.3\textwidth]{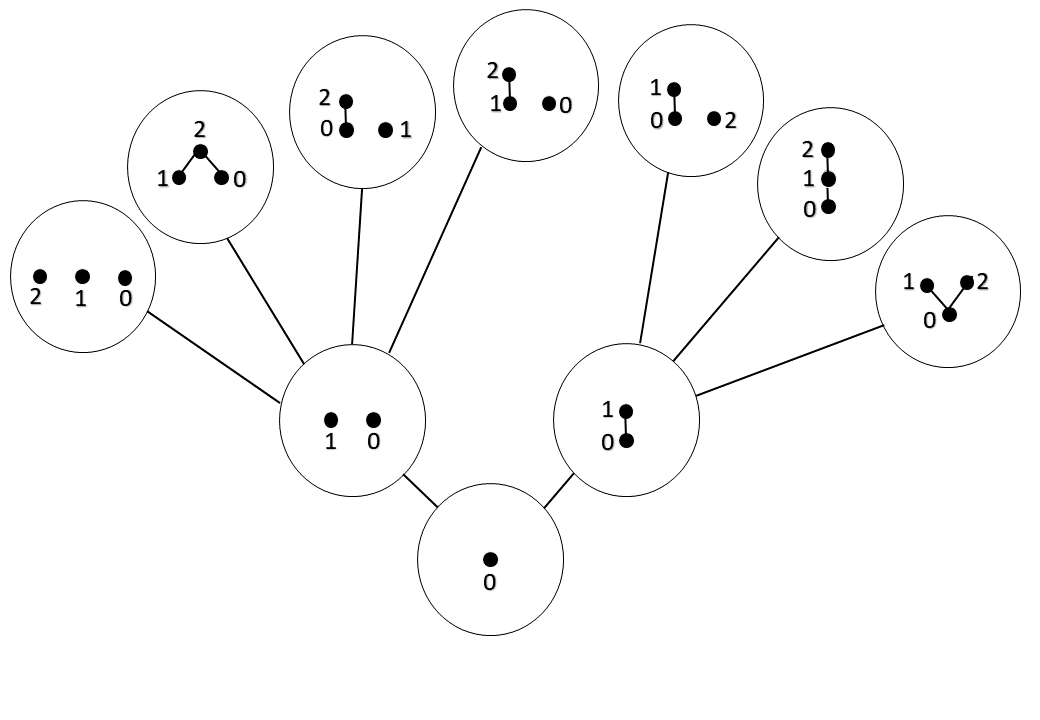}
  \caption[Labeled poscau]{The first three levels of labeled poscau.}\label{labeled_poscau}
\end{figure}

The algebra of \textit{covariant} events, $\mathcal{R}$, is a subalgebra of $\lc{\mathcal{R}}$: 
\begin{align}\mathcal{R}:=\{\mathcal{E}\in \lc{\mathcal{R}}|\lc{C}\in \mathcal{E} \text{ and } \lc{C}\cong \lc{C}' \implies \lc{C}'\in \mathcal{E}\},\end{align} where $\cong$ denotes equivalence under order-isomorphism. 

Within this framework for the dynamics of a discrete universe, the events in $\mathcal{R}$ are the physical ``observables'' (or beables)\cite{Brightwell:2002vw}. Each event $\mathcal{E}$ in $\mathcal{R}$ corresponds to the question ``Does $\mathcal{E}$ happen?''  to which the measure responds: ``Yes, with probability $\mu(\mathcal{E})$'' (or ``Almost surely no'' if $\mu(\mathcal{E})=0$). 

\subsection{Rogues}

Looking closely,  one finds that a generic event in $\mathcal{R}$ has no obvious physical interpretation. 
However there exists a strictly smaller $\sigma$-algebra whose elements do have a clear physical meaning. Let $C_n$ denote an \textit{unlabeled causet} (or \textit{order}), \textit{i.e.}  the order-isomorphism equivalence class of which the causet $\lc{C}_n$ is a representative. For each $C_n$, define the stem-set,
\begin{align}\label{18073}
stem(C_n):=\bigcup cyl(\lc{D}_m),
\end{align}
where the union is over causets $\lc{D}_m$ that  contain a stem that  is order-isomorphic to $\lc{C}_n$, for all $m$ (see Fig.\ref{stemfig} for an illustration of stem).
Let $\mathcal{S}$ and $\mathcal{R}(\mathcal{S})$ denote the set of stem-sets for all $n$ and the $\sigma$-algebra generated by $\mathcal{S}$, respectively. The elements of $\mathcal{S}$ have a comprehensible physical meaning: they correspond to countable logical combinations of statements like ``the causet has a stem 
isomorphic to a representative of the finite order $C_n$.''  $\mathcal{R}(\mathcal{S})$ is the prime example of a physically comprehensible subalgebra of the covariant event algebra $\mathcal{R}$.

\begin{figure}[htpb]
\centering
    \includegraphics[width=.3\textwidth]{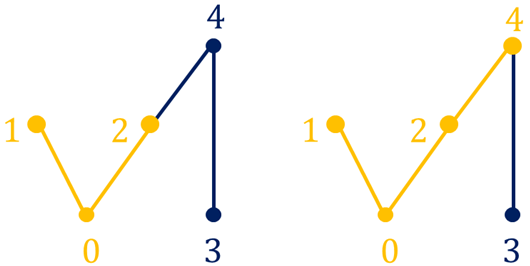}
  \caption{An illustration of the concept of stem. Left: the sub-causet highlighted in yellow is a stem (down-set) because it contains its own past. Right: the sub-causet highlighted in yellow is \textit{not} a stem because it contains the element 4 but it does not contain the element 3 which lies to the past of 4.}\label{stemfig}
\end{figure}

What physics does $\mathcal{R}(\mathcal{S})$ leave out, what physical information is not captured by the stem questions? The answer, almost by definition, is that stem questions cannot distinguish between two causets  that have the same stems. We call an infinite causet $\lc{U}\in\lc{\Omega}$ a \textit{rogue} if there exists some $\lc{V}\not\cong \lc{U}$ such that $\lc{V}\in stem(C_n)$ if and only if $\lc{U}\in stem(C_n)$, for all $stem(C_n)\in\mathcal{S}$. Let $\Theta$ denote the set of rogues.  
In \cite{Brightwell:2002yu, Brightwell:2002vw} it was proved that $\Theta$ is measurable (\textit{i.e.} $\Theta\in\mathcal{R}$) and indeed that $\Theta\in\mathcal{R(S)}$. The main theorem of \cite{Brightwell:2002vw} is that
 \textit{for every event $\mathcal{E}\in\mathcal{R}$, there is an event $\mathcal{E}'\in\mathcal{R}(\mathcal{S})$ such that $\mathcal{E}\triangle\mathcal{E}'\subset\Theta$}. 
It is in this precise technical sense that the rogues make up the difference between the covariant events and the stem events (elements of the stem event algebra 
$\mathcal{R}(\mathcal{S})$).

We can go a little further and prove 
\begin{lemma} \label{lemma0}
Let $ \mathcal{E}\in\mathcal{R}$. Then $\mathcal{E} \cap \Theta^c \in \mathcal{R}(\mathcal{S})$ and $\mathcal{E} \cup \Theta \in \mathcal{R}(\mathcal{S})$, 
where the superscript ${}^c$ denotes complement in $\lc{\Omega}$. 
\end{lemma} 
\begin{proof} 
Define $\mathcal{F}:= \mathcal{E} \cap \Theta^c \in \mathcal{R}$. By 
the theorem mentioned above there is some $\mathcal{F}' \in 
\mathcal{R}(\mathcal{S})$ such that $(\mathcal{F}\setminus \mathcal{F}' )\cup (\mathcal{F}'\setminus \mathcal{F}) \subset\Theta$. This implies $ \mathcal{F}'\setminus \mathcal{F} \subset\Theta$. Since 
$\mathcal{F} \cap  \Theta = \emptyset$, we have 
$\mathcal{F} = \mathcal{F}'\setminus \Theta$. And since $\Theta \in \mathcal{R}(\mathcal{S})$, 
we have $\mathcal{E} \cap \Theta^c \in \mathcal{R}(\mathcal{S})$.\\

Now define  $\mathcal{G}:= \mathcal{E} \cup \Theta \in \mathcal{R}$. Then there is a $\mathcal{G}' \in 
\mathcal{R}(\mathcal{S})$ such that $(\mathcal{G}\setminus \mathcal{G}' )\cup (\mathcal{G}'\setminus \mathcal{G}) \subset\Theta$. This implies $\mathcal{G}\setminus \mathcal{G}' \subset \Theta$. And so 
$\mathcal{G} = \mathcal{G}'\cup \Theta$. And since $\Theta \in \mathcal{R}(\mathcal{S})$, 
we have $\mathcal{E} \cup \Theta \in \mathcal{R}(\mathcal{S})$.\\

\end{proof}
Thus, removing the rogues from a covariant event turns it into a stem event. And adding all the rogues to a covariant event turns it into a stem event. This motivates the defining of another algebra that will have a direct analogue when we come to discuss cyclic dynamics in the next section. Let $\mathcal{R}_{\Theta}$ be the 
 the $\sigma$-algebra of all covariant events that either contain all the rogues or contain no rogues:
 \begin{align}\label{roguealgebra}\mathcal{R}_{\Theta}:=\{\mathcal{E}\in\mathcal{R}\,|\,\Theta\subseteq\mathcal{E} \text{ or } \Theta\subseteq\mathcal{E}^c\}\,.\end{align}
  In any dynamics in which the rogues have measure zero, i.e. 
 $\mu(\Theta)=0$, 
$\mathcal{R}_{\Theta} $ is then a sort of doubled physical event algebra 
 \begin{align}\label{roguealgebra2}\mathcal{R}_{\Theta}=\left(\mathcal{R}\cap \Theta^c \right)
 \sqcup \left(\mathcal{R}\cup \Theta\right)
\end{align} where  
 \begin{align}\label{roguealgebra3}\mathcal{R}\cap \Theta^c : = &\{ \mathcal{E} \cap \Theta^c \,|\,\mathcal{E}\in\mathcal{R}\}\\
 \mathcal{R}\cup \Theta : = &\{ \mathcal{E} \cup \Theta \,|\,\mathcal{E}\in\mathcal{R}\}\,.
\end{align}

\subsection{CSG models, a special class}

Crucially, it was also proved that in every Classical Sequential Growth (CSG) model, the most-studied class of sequential growth dynamics  \cite{Rideout:1999ub},  the set of rogues $\Theta$ has measure zero  \cite{Brightwell:2002yu, Brightwell:2002vw}.  Combined with the measure independent 
theorem stated above, this means that in every CSG model, for every covariant event there is a stem event such that their difference is of measure zero. 
Indeed, the lemma proved above says that removing the rogues from a covariant event turns it into a stem event and 
in dynamics satisfying $\mu(\Theta)=0$, such as CSG models, 
$\mu(\mathcal{E}) = \mu(\mathcal{E}\cap\Theta^c)$ for any covariant event $\mathcal{E}$. 
So, for CSG models, a physically motivated class of sequential growth dynamics, the stem observables exhaust the physical observables in this well-defined sense. 

Finally we can summarise the relations between the $\sigma$-algebras mentioned so far: 
\begin{align} \mathcal{R}\cap \Theta^c \subset \mathcal{R}_{\Theta}\subset \mathcal{R}(\mathcal{S})\subset\mathcal{R}\,.\end{align}  $\mathcal{R}(\mathcal{S})$ is, in this sense, an ``over-complete'' set of observables in dynamics in which the set of rogues has measure zero.  It is the simple physical interpretation of $\mathcal{R}(\mathcal{S})$ that makes it the most meaningful choice of physical event algebra, or its elements the most meaningful choice of physical observables. 

The significance of the results about $\mathcal{R}(\mathcal{S})$  is two-fold. First, as mentioned above, we can now assign a clear meaning to every dynamically relevant observable in dynamics satisfying $\mu(\Theta)=0$: it is a logical combination of statements about which (unlabeled) stems are contained in the causet spacetime. For example, the statement ``the causal set universe has a unique minimal element'' corresponds to a stem event because it is equivalent to the statement ``the causal set universe does not contain the 2-antichain as a stem''. This event is particularly interesting for causal set cosmology, as it can be interpreted as a Big Bang event. Second, this result has led to the construction of \textit{covtree}---a tree on which every random walk corresponds directly to a measure on $\mathcal{R}(\mathcal{S})$ (and vice versa)---a new tool for studying causal set dynamics \cite{Dowker:2019qiz,Zalel:2020oyf,Bento:2021voo}. More generally, a
better understanding of the structure and sub-algebras of $\mathcal{R}$ can enable us to develop new methods through which to define a measure $\mu$ and thus to establish new avenues for seeking physically motivated dynamics.

\section{Cyclic models and observables from principal cylinder sets}\label{chap_6_o_subsec}
 
 \subsection{Cyclic models}
 Our theory space is that of the sequential growth models of section \ref{sec:sequential} of which
CSG models are a special case. Our theorems make no use of the particular properties of CSG models but we have them in mind as a physically motivated class of models to which we can apply the theorems. 

Let us call a causet with infinitely many breaks, a \textit{cyclic} causet and let $\mathcal{B}_{\infty}$ denote the event that the causet is cyclic, \textit{i.e.},
\begin{equation}
\mathcal{B}_{\infty}:=\{ \lc{C}\in\lc{\Omega}|\lc{C} \text{ contains infinitely many breaks}  \}.\end{equation}

\begin{lemma}\label{lemma030523}
$\mathcal{B}_\infty \in \mathcal{R}(\mathcal{S})$.
\end{lemma}
\begin{proof} 
Let event $\Gamma_n(A_n)$ be the event ``the causet has a break with past $A_n$'' where $A_n$ is an $n$-order. In  \cite{Zalel:2020oyf} it was proved that the event $\Gamma_n(A_n)$ is a stem event: 
\begin{align}
 \Gamma_n(A_n) = stem(\hat{A_n}) \bigcap_{X_{n+1} \ne \hat{A_n}} stem(X_{n+1})^c\,,
 \end{align}
 where $\hat{A_n}$ is the covering order of $A_n$, the $(n+1)$-order formed by adding a single element that is above every element of $A_n$,  and where the intersection is over all $(n+1)$-orders not equal to $\hat{A_n}$.  Now let event $\Gamma_n$ be the event ``the causet has a break the cardinality of whose past is $n$'' ($n >0$ by definition of break). 
\begin{align}
 \Gamma_n =  \bigcup _{A_{n} } \Gamma_n(A_n) \,,
 \end{align} 
where the union is over all $n$-orders. 

 $\mathcal{B}_\infty^c$ is the event that the causet has finitely many breaks. A causet has finitely many breaks if there exists a $k \in \mathbb{N}$ such that it has no break whose past has cardinality greater than $k$. 
 Let $\Delta_k$ be the event that the causet has no break whose past has cardinality greater than $k$. 
 \begin{align}
 \Delta_k =  \bigcap _{ n=k+1}^\infty \Gamma_n^c \,. 
 \end{align} 
 And then
 \begin{align}
\mathcal{B}_\infty^c = \bigcup_{k = 0}^\infty \Delta_k \,. 
 \end{align} 
Hence, $\mathcal{B}_\infty$ and its complement are elements of $\mathcal{R}(\mathcal{S})$. 
\end{proof}

Since $\mathcal{R}(\mathcal{S})\subset\mathcal{R}$, every event in  $\mathcal{R}(\mathcal{S})$ is measureable in any sequential growth model. Therefore, it is a corollary of lemma \ref{lemma030523} that $\mathcal{B}_\infty$ is measureable in any sequential growth model. This result enables us to define a cyclic model as follows: a cyclic model is a sequential growth model in which $\mu(\mathcal{B}_{\infty})=1$.

A cyclic model may or may not be a CSG model and a CSG model may or may not be cyclic, though the best-understood growth dynamics---Transitive Percolation---is both, since it is a CSG model  in which the causet spacetime almost surely has infinitely many posts \cite{Bombelli:2008kr,Alon:1994,Ash:2002un,Ash:2005za,Brightwell:2016}.

Our motivation for studying cyclic models is the key role that they play in the the causal set comological paradigm that  aims to explain the emergence of a flat, homogeneous and isotropic cosmos directly from the quantum gravity era \cite{Sorkin:1998hi,Martin:2000js,Dowker:2017zqj}. Our goal in this paper, is to classify the observables in these models.

Now, in analogy to definition \ref{roguealgebra}, we define $\mathcal{R}_b$ to be the $\sigma$-algebra of all covariant events that either contain all or none of the causets which are \textit{not} cyclic:
 \begin{equation}\label{18075}\begin{split} \mathcal{R}_b:=&\{\mathcal{E} \in\mathcal{R}|\mathcal{B}_{\infty}^c\subseteq \mathcal{E}\text{ or }\mathcal{B}_{\infty}^c \subseteq \mathcal{E}^c\}\\
 = & \left(\mathcal{R}\cap \mathcal{B}_\infty \right)
 \sqcup \left(\mathcal{R}\cup \mathcal{B}_\infty^c\right),
 \end{split}\end{equation}
 where $\mathcal{B}_{\infty}^c$ denotes the complement of $\mathcal{B}_{\infty}$, \textit{i.e.} the set of causets which are \textit{not} cyclic. 
It follows from our definition of cyclic models that $\mu(\mathcal{B}_{\infty}^c)=0$ in any cyclic model, and therefore $\mathcal{R}_b$ exhausts the covariant events in a cyclic model in the same way that $\mathcal{R}_{\Theta}$ does in dynamics which satisfy $\mu(\Theta)=0$.
Note,  however that the events in $\mathcal{R}_b$ are not necessarily stem events and so do not in general have a physical interpretation. In section \ref{sec2} we will prove a theorem in exactly analogous form to the result of \cite{Brightwell:2002vw} and identify the covariant, physically interpretable observables in a cyclic dynamics.

\subsection{$\mathcal{R}_b$ is generated by principal cylinder sets}

In this subsection, we prove that $\mathcal{R}_b$ contains exactly the covariant events in the sigma algebra that is generated by a strict subset of the set of all cylinder sets.
 
The following terminology will be useful. We call $\lc{C}_n$ a \textit{principal causet} if (i) it is a causet on ground-set $[0,n-1]$ satisfying $x\prec y\implies x<y$, and (ii) it has a unique maximal element. If a principal causet $\lc{C}_n$ is a subcauset of some $\lc{C}\in\lc{\Omega}$ then $\lc{C}_n$ is a stem in $\lc{C}$ and we say that $\lc{C}_n$ is a \textit{principal stem} in $\lc{C}$. We call  $cyl(\lc{C}_n)$ a \textit{principal cylinder set} if $\lc{C}_n$ is a principal causet. Let $\lc{\mathcal{R}}_b$ denote the $\sigma$-algebra generated by the set of all principal cylinder sets.

As we will prove below, the notion of principal causets and principal cylinder sets is closely related to the notion of observables in cyclic dynamics. At this stage, we can already develop an intuition for why this is the case. Suppose the left hand diagram of Fig.\ref{postbreak} represents a break $(\lc{A},\lc{B})$ in a cyclic causet $\lc{C}$, and let $n:=|\lc{A}|$. Then element $n$ is a minimal element in $\lc{B}$ and the causet $\lc{A}\cup\{n\}$ is a principal causet and a principal stem in $\lc{C}$. One can also see that every stem of cardinality $n+1$ in $\lc{C}$ is isomorphic to $\lc{A}\cup\{n\}$. So, every cyclic causet $\lc{C}$ contains countably many principal stems — one for each break — and every stem in $\lc{C}$ is contained within some \textit{principal} stem in $\lc{C}$. The intuition is that all the physical information about $\lc{C}$ is encoded in the principal stems it contains. In the rest of this section, we make this intuition mathematically precise.

We can now state our theorem:
\begin{theorem}\label{theorem_sec_6_1}$\mathcal{R}_b=\lc{\mathcal{R}}_b\cap{\mathcal{R}}$\,,\end{theorem}
\noindent and we spend the rest of this subsection proving it.

We will use the following terminology when discussing breaks. Given a break  $(\lc{A},\lc{B})$, $\lc{A}$ and $\lc{B}$ are called the \textit{past} and the \textit{future} of the break, respectively. If a causet $\lc{C}$ contains more than one break we order its breaks by the cardinality of their pasts: $(\lc{A}_1,\lc{B}_1),(\lc{A}_2,\lc{B}_2),...$ where $|\lc{A}_1|<|\lc{A}_2|<...$. 
The pasts of the breaks are a nested increasing sequence of stems (finite down sets): $\lc{A}_1 \subset \lc{A}_2
\subset \lc{A}_3 \dots$. The futures of the breaks are a nested decreasing sequence of infinite up sets:
 $\lc{B}_1 \supset \lc{B}_2
\subset \lc{B}_3 \dots$.

A \textit{segment} (or \textit{epoch}) is a subcauset that  lies between two consecutive breaks (\textit{i.e.} the segments of $\lc{C}$ are $A_1, A_2\setminus A_1, A_3\setminus A_2,...$). If $\lc{C}$ contains exactly $k$ breaks where $0<k<\infty$ then $B_k$ is also a segment, the infinite final epoch. If $\lc{C}$ contains no breaks then $\lc{C}$ itself is the only segment in $\lc{C}$. 

Let $\lc{\mathcal{F}}\subset\lc{\Omega}$ denote the set of all infinite causets that  contain finitely many principal causets as stems (\textit{i.e.} $\lc{C}\in\lc{\mathcal{F}} \iff \exists \  m\in\mathbb{N}$ such that, for all $n>m$, $\lc{C}|_{[0,n]}$ contains at least two maximal elements).

\begin{lemma}\label{claim100} $\lc{C}\in \mathcal{B}_{\infty}^c$ if and only if there exists some $\lc{D}\in\lc{\mathcal{F}}$ such that $\lc{D}\cong \lc{C}$. \end{lemma}

\begin{proof}
Note that $\lc{\mathcal{F}}\subset\mathcal{B}_{\infty}^c$, since a causet with infinitely many breaks necessarily contains infinitely many principal causets as stems. Therefore, if there exists some $\lc{D}\in\lc{\mathcal{F}}$ such that $\lc{D}\cong \lc{C}$ then $\lc{C}\in \mathcal{B}_{\infty}^c$.

To prove the converse, let $\lc{C}\in \mathcal{B}_{\infty}^c$ and let $N$ be the cardinality of the past of the last break in $\lc{C}$. Suppose $\lc{C}$ contains infinitely many principal causets as stems and consider the infinite sequence, $$n_1<n_2<...\equiv (n_i)$$ of integers $n_i>N+1$ for which the restriction $\lc{C}|_{[0,n_i]}$ is a principal causet (\textit{i.e.} $n_i\succ x \ \forall \ x<n_i$).

We now construct a bijection $g:\lc{C}\rightarrow \mathbb{N}$ and define $\lc{D}\in\lc{\Omega}$ to be the infinite causet in which $x\prec y \iff g^{-1}(x)\prec g^{-1}(y)$ in $\lc{C}$. The definition of $\lc{D}$ ensures that $\lc{D}\cong\lc{C}$  and our construction of $g$ ensures that $\lc{D}\in\lc{\mathcal{F}}$. This proves the claim.

Set $g(x)=x \ \forall \ x<n_1$. Set $g(n_1)=n_1+1$ (this ensures that $\lc{D}|_{[0,n_1]}$ will not be a principal causet). Let  $n_1^*$ denote the smallest integer greater than $n_1$ that  satisfies both $n_1^*\not\succ n_1$ and $n_1^*\not\succ y$ for at least one $y<n_1-1$ in $\lc{C}$. (Such an integer surely exists because otherwise $\lc{C}$ would contain a break with past $[0,n_1-1]$, which is a contradiction.) Set $g(x)=x+1 \ \forall \ n_1<x<n_1^*$ and $g(n_1^*)=n_1$.

Now, define $m_2$ to be equal to the smallest entry in $(n_i)$ that  is larger than $n_1^*$. Set $g(x)=x \ \forall \  n_1^*<x<m_2$. Set $g(m_2)=m_2+1$ (this ensures that $\lc{D}|_{[0,m_2]}$ will not be a principal causet). Let  $n_2^*$ denote the smallest integer greater than $m_2$ that  satisfies both $n_2^*\natural m_2$ and $m_2^*\not\succ y$ for at least one $y<m_2-1$ in $\lc{C}$. Set $g(x)=x+1 \ \forall \ m_2<x<n_2^*$ and $g(n_2^*)=m_2$. Repeat this process with $m_2\rightarrow m_3,n_2^*\rightarrow n_3^*$ \textit{etc.}
\end{proof}

\begin{figure}[htpb]
  \centering
	\includegraphics[width=0.48\textwidth]{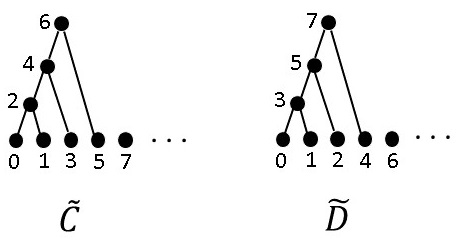}
	\caption[Labeled causets without breaks]{Illustration of lemma \ref{claim100}. $\lc{C}\in\mathcal{B}_{\infty}^c$, since it contains no breaks. $\lc{C}\not\in\lc{\mathcal{F}}$, since $\lc{C}|_{[0,2n]}$ is a principal causet $\forall \ n\in\mathbb{N}$. $\lc{D}\in\lc{\mathcal{F}}$, since $\lc{D}|_{[0,n]}$ is \textit{not} a principal causet $\forall \ n>0$.}
	\label{relabel_non_break} 
\end{figure}

\begin{corollary}\label{cor3} If $\mathcal{E}\in\mathcal{R}$ and $\lc{\mathcal{F}}\subset \mathcal{E}$ then $\mathcal{B}_{\infty}^c\subseteq \mathcal{E}$.
\end{corollary}

\begin{lemma}\label{claim11} Let $\mathcal{E} \in\lc{\mathcal{R}}_b\cap{\mathcal{R}}$. $\mathcal{E}\cap\mathcal{B}_{\infty}^c\not=\emptyset \implies \mathcal{B}_{\infty}^c\subseteq \mathcal{E}.$  \end{lemma}

\begin{proof} For each principal causet $\lc{C}_n$ let $S(\lc{C}_n)$ denote the set of principal causets with cardinality greater than $n$ whose restriction to $[0,n-1]$ is $\lc{C}_n$, and define \begin{align} \Gamma_{\lc{C}_n}:=cyl(\lc{C}_n)\setminus\bigcup_{\lc{D}_m\in S(\lc{C}_n)}cyl(\lc{D}_m).\end{align}

\noindent $\Gamma_{\lc{C}_n}$ is the set of infinite causets that  (i) contain the principal causet $\lc{C}_n$ as a stem and (ii) contain no principal causet of cardinality greater than $n$ as a stem. We will use the following properties:

\begin{enumerate}
\item Each $\Gamma_{\lc{C}_n}$ is an atom of $\lc{\mathcal{R}_b}$ (\textit{i.e.} the elements of  $\Gamma_{\lc{C}_n}$ cannot be separated by the principal cylinder sets).
\item The collection of all the $\Gamma_{\lc{C}_n}$ is a partition of $\lc{\mathcal{F}}$ (since every $\lc{C}\in\lc{\mathcal{F}}$ is contained in some $\Gamma_{\lc{C}_n}$ and for any two principal causets $\lc{C}_n\not=\lc{D}_m$ we have $\Gamma_{\lc{C}_n}\cap \Gamma_{\lc{D_m}}=\emptyset$). 
\end{enumerate}

Given a principal causet $\lc{C}_n$, let $\lc{X}_{\lc{C}_n}\in\lc{\Omega}$ denote the infinite causet whose restriction to $[0,n-1]$ is $\lc{C}_n$ and in which all elements $m\geq n$ are unrelated to all others. Let $\lc{X}'_{\lc{C}_n}$ denote a causet isomorphic to $\lc{X}_{\lc{C}_n}$ in which the element 0 is unrelated to all others. Then $\lc{X}_{\lc{C}_n}\in \Gamma_{\lc{C}_n}$ and $\lc{X}'_{\lc{C}_n}\in \Gamma_{\lc{C}_1}$.

\vspace{1mm}
Suppose $\Gamma_{\lc{C}_1}\subset \mathcal{E}\in \lc{\mathcal{R}}_b\cap{\mathcal{R}}$. Since $\Gamma_{\lc{C}_1}\subset \mathcal{E}$ we have $\lc{X}'_{\lc{C}_n}\in \mathcal{E}$ for all principal causets $\lc{C}_n$. Since $\mathcal{E}$ is covariant, $\lc{X}_{\lc{C}_n}\in \mathcal{E}$ for every principal $\lc{C}_n$. Hence, by property 1, $\Gamma_{\lc{C}_n}\subset \mathcal{E}$ for every principal $\lc{C}_n$. By property 2, $\lc{\mathcal{F}}\subset \mathcal{E}$ and by corollary \ref{cor3}, $\mathcal{B}_{\infty}^c\subseteq \mathcal{E}$.

Now, consider any $\mathcal{E}\in \lc{\mathcal{R}}_b\cap{\mathcal{R}}$ for which $\mathcal{E}\cap \mathcal{B}_{\infty}^c\not=\emptyset$ and let $\lc{C}\in \mathcal{E}\cap \mathcal{B}_{\infty}^c$. Then there exists some $\lc{D}\cong\lc{C}$ and some $\lc{C}_n$ such that $\lc{D}\in \Gamma_{\lc{C}_n}$ and therefore $\Gamma_{\lc{C}_n}\subset \mathcal{E}$. Hence $\lc{X}_{\lc{C}_n}\in \Gamma_{\lc{C}_n}$. Since $\mathcal{E}$ is covariant, $\lc{X}'_{\lc{C}_n}\in \mathcal{E}$. Therefore $\Gamma_{\lc{C}_1}\subset \mathcal{E}$, which completes the proof.
\end{proof}

\begin{corollary}\label{18071}$\lc{\mathcal{R}}_b\cap{\mathcal{R}}\subseteq\mathcal{R}_b.$\end{corollary}

\begin{lemma}\label{lemma_09062103}$\mathcal{B}_{\infty}\in\lc{\mathcal{R}_b}$. \end{lemma}
\begin{proof}

Consider the collection of $\lc{C}_n$ that  contain no breaks. We enumerate these causets using the label $i\in\mathbb{N}$, so that $\lc{C}_{n_{i}}$ is the $i^{th}$ causet that  contains no breaks and its cardinality is $n_{i}$. We will use the string $\lc{C}_{n_{i_1}}...\lc{C}_{n_{i_{k-2}}}{\lc{C}_{n}{_{i_{k-1}}}}$ to represent the finite principal causet which has $k$ segments and whose $j^{th}$ segment is (canonically order-isomorphic to) $\lc{C}_{n_{i_j}}$ for $1\leq j\leq k-1$.

For each causet $\lc{C}_{n_{i_1}}$ that  contains no breaks, define the set,
\begin{align}\label{0906211_eq} \mathcal{B}_1(\lc{C}_{n_{i_1}}):=\bigcap_{k=2}^{\infty} \bigcup cyl(\lc{C}_{n_{i_1}}\lc{C}_{n_{i_2}}...\lc{C}_{n_{i_{k-1}}}{\lc{C}_{{n}_{i_{k}}}}),\end{align}
where the union is over all sequences $(i_2,....,i_k)$ of natural numbers.
Note that if $\lc{C}\in\mathcal{B}_1(\lc{C}_{n_{i_1}})$ then $\lc{C}_{n_{i_1}}$ is the first segment in $\lc{C}$, and therefore $\lc{C}$ contains at least one break. Additionally, if $\lc{C}\in\mathcal{B}_{\infty}$ and  $\lc{C}_{n_{i_1}}$ is the first segment in $\lc{C}$ then $\lc{C}\in\mathcal{B}_1(\lc{C}_{n_{i_1}})$. 

We now generalise \eqref{0906211_eq} to any $l\geq 1$. Given a string $\lc{C}_{n_{i_1}}...\lc{C}_{n_{i_l}}$ of finite causets that  contain no breaks define the set
\begin{align} \mathcal{B}_l(\lc{C}_{n_{i_1}}...\lc{C}_{n_{i_l}}):=\bigcap_{k=l+1}^{\infty} \bigcup cyl(\lc{C}_{n_{i_1}}...\lc{C}_{n_{i_l}} \lc{C}_{n_{i_{l+1}}}...{\lc{C}_{n_{i_{k}}}})\end{align}
where the union is over all sequences $(i_{l+1},....,i_{k})$ of natural numbers. 
If $\lc{C}\in\mathcal{B}_l(\lc{C}_{n_{i_1}}...\lc{C}_{n_{i_l}})$ then $\lc{C}_{n_{i_1}},...,\lc{C}_{n_{i_l}}$ are the first $l$ segments in $\lc{C}$ (and therefore $\lc{C}$ contains at least $l$ breaks). Additionally, if $\lc{C}\in\mathcal{B}_{\infty}$ and  $\lc{C}_{n_{i_1}},...,\lc{C}_{n_{i_l}}$ are the first $l$ segments in $\lc{C}$ then $\lc{C}\in\mathcal{B}_l(\lc{C}_{n_{i_1}}...\lc{C}_{n_{i_l}})$.
Define,
\begin{align} \mathcal{B}_l:= \bigcup\mathcal{B}_l(\lc{C}_{n_{i_1}}\lc{C}_{n_{i_2}}...\lc{C}_{n_{i_l}}),\end{align}
where the union is over all sequences $(i_1,...,i_l)$.
For each $l$, $\mathcal{B}_l\supset\mathcal{B}_\infty$ and if $\lc{C}\in\mathcal{B}_{l}$ then $\lc{C}$ contains at least $l$ breaks.
Therefore,
\begin{align} \mathcal{B}_{\infty}= \bigcap_{l=1}^{\infty}\mathcal{B}_l.\end{align}
\end{proof}

\begin{lemma}\label{18072}$\mathcal{R}_b\subseteq\lc{\mathcal{R}}_b\cap{\mathcal{R}}.$\end{lemma}
\begin{proof}
First, we show that if $\mathcal{E}\in\mathcal{R}$ and $\mathcal{E}\subseteq \mathcal{B}_{\infty}$ then $\mathcal{E}\in\lc{\mathcal{R}}_b\cap{\mathcal{R}}$.

Let $\mathcal{E}\in\mathcal{R}$ and $\mathcal{E}\subseteq \mathcal{B}_{\infty}$. Then $\Theta\in\mathcal{E}^c$ and therefore $\mathcal{E}\in\mathcal{R}(\mathcal{S})$. Moreover, $\mathcal{E}$ is in the restriction of $\mathcal{R}(\mathcal{S})$ to $\mathcal{B}_{\infty}$, where the restriction is defined by,
 \begin{align} 
\mathcal{R}(\mathcal{S})\cap\mathcal{B}_{\infty}:=\{ \mathcal{G} \cap \mathcal{B}_{\infty} |\mathcal{G} \in\mathcal{R}(\mathcal{S})\},
 \end{align}
 and is countably generated by the collection,
  \begin{align}
  \mathcal{S}\cap\mathcal{B}_{\infty}:= \{stem(C_n)\cap \mathcal{B}_{\infty}|stem(C_n)\in\mathcal{S}\}.
  \end{align} 
  Using definition \eqref{18073} we can write, \begin{align}\label{18074} stem(C_n)\cap \mathcal{B}_{\infty}=\big(\bigcup cyl(\lc{D}_m)\big)\cap \mathcal{B}_{\infty},\end{align} where the union is over all causets $\lc{D}_m$ for all $m$ that  contain a stem that  is order-isomorphic to a representative of $C_n$. Now, note that any causet in $cyl(\lc{D}_m)\cap \mathcal{B}_{\infty}$ contains infinitely many breaks and therefore has a principal stem of cardinality $>m$ that contains $\lc{D}_m$ as stem. Therefore, we can restrict the domain of the union in \eqref{18074} to the set of \textit{principal} causets $\lc{D}_m$ for all $m$ that  contain a stem that  is order-isomorphic to a representative of $C_n$. Combined with lemma \ref{lemma_09062103}, this proves that every set of the form \eqref{18074} is contained in $\lc{\mathcal{R}}_b$ and hence $\mathcal{E}\in\lc{\mathcal{R}}_b$. Since $\mathcal{E}$ is covariant by assumption, $\mathcal{E}\in\lc{\mathcal{R}}_b\cap{\mathcal{R}}$.

Now, let $\mathcal{A}\in\mathcal{R}_b$. By definition \eqref{18075}, either $\mathcal{B}_{\infty}^c\subseteq \mathcal{A}$ or $\mathcal{B}_{\infty}^c \subseteq \mathcal{A}^c$. In the latter case, $\mathcal{B}_{\infty}^c \subseteq \mathcal{A}^c \implies \mathcal{A}\subseteq \mathcal{B}_{\infty} \implies  \mathcal{A}\in\lc{\mathcal{R}}_b\cap{\mathcal{R}}$. In the former case, $\mathcal{B}_{\infty}^c\subseteq \mathcal{A} \implies \mathcal{A}^c \subseteq \mathcal{B}_{\infty} \implies \mathcal{A}^c \in\lc{\mathcal{R}}_b\cap{\mathcal{R}}\implies \mathcal{A}\in\lc{\mathcal{R}}_b\cap{\mathcal{R}}$.\end{proof}

Corollary \ref{18071} and lemma \ref{18072} together imply theorem \ref{theorem_sec_6_1}.

 \subsection{Cyclic models as random walks up a new tree}\label{subsec:reducedposc} 

In the context of random walks, the upshot of theorem \ref{theorem_sec_6_1} is that one can conceive of the cyclic models that  satisfy $\mu(\mathcal{B}_{\infty})=1$ as random walks up a new tree which we dub \textit{reduced poscau} (Fig. \ref{reduced_poscau}).
\begin{figure}[htpb]
\centering
    \includegraphics[width=.48\textwidth]{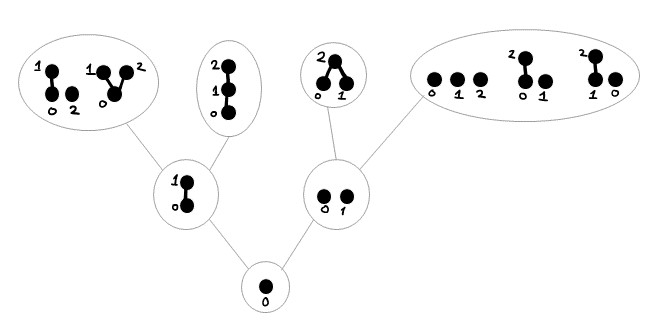}
  \caption[Reduced poscau]{The first three levels of reduced poscau.}\label{reduced_poscau}
\end{figure}

One can think of reduced poscau as obtained from labeled poscau by merging groups of nodes into one, so that in reduced poscau each node is a \textit{collection} of causets. Principal causets are never merged with others, so each principal causet is contained in a node on its own. Given a principal causet, all the causets that  are directly above it in labeled poscau (except for the one that  is itself a principal causet) are merged into one node. Thus, in reduced poscau each principal causet has exactly two nodes directly above it. A node that  contains $r$ causets has $r+1$ nodes above it: $r$ nodes, each of which contains a single principal causet (formed by adding a single element above each of the $r$ causets), and an additional node that  contains all other causets that  in labeled poscau are directly above any of the $r$ causets. The meaning of each node is ``one of these causets is a stem in the growing causet''.

There is a correspondence between the nodes of reduced poscau and sets of $\lc{\Omega}$. A node that  contains a principal causet $\lc{C}_n$ corresponds to the principal cylinder set $cyl(\lc{C}_n)$. A node that  contains non-principal causets corresponds to a set that  is constructed recursively from principal cylinder sets as follows. Let $a$ be a non-principal node directly above the node $b$ and let $\lc{C}_{n}^1,...,\lc{C}_{n}^r$ be the principal causets contained in the remaining nodes directly above $b$. Then the set corresponding to $a$ is equal to the set corresponding to $b$ take away $\bigcup_{i=1}^r cyl(\lc{C}_{n}^i)$. If $b$ is a non-principal node then the set corresponding to it can be constructed in the same manner. One works recursively down the path and the process ends when reaching a principal node.

It follows that the $\sigma$-algebra generated by reduced poscau via this correspondence is equal to $\lc{\mathcal{R}_b}$. The collection of all sets corresponding to nodes in reduced poscau is a semi-ring, and therefore each random walk on reduced poscau induces a probability measure on ${\lc{\mathcal{R}}_b}$, where the measure of each set in the semi-ring is equal to the probability of reaching the corresponding node. Thus, reduced poscau can be used to define cyclic dynamics. Note that the set of principal cylinder sets alone does not form a semi-ring and therefore we cannot work with it directly to define a measure on ${\lc{\mathcal{R}}_b}$. The sets that  correspond to non-principal nodes complete the collection into a semi-ring and allow to define a measure.

Note that each walk on reduced poscau induces a walk on labeled poscau (the proof is analogous to that of lemma 4.9 in \cite{Dowker:2019qiz}) so that the walks on labeled poscau and on reduced poscau yield the same class of probability measures on the covariant $\mathcal{R}_b$. However, the different structures of the two trees could lead to different formulations of physical constraints in terms of transition probabilities and hence to identifying different classes of physically interesting dynamics.

\subsection{Complex Transitive Percolation}\label{sec:CTP}
The formulation of causal set dynamics in terms of probability measures and random walks is a precursor to the fully quantum dynamics to be expressed as a decoherence functional. In \citep{Dowker:2010qh} it was proposed that a decoherence functional can be obtained from a complex measure on $\lc{\mathcal{R}}$, itself derived as follows: replace the real transition probabilities on labeled poscau by complex transition amplitudes $A(\lc{C}_n\rightarrow\lc{C}_{n+1})$ that  satisfy the sum-rule,
\begin{align}\label{eq_sumrule_sec5}
\sum_i A(\lc{C}_n\rightarrow \lc{C}^i_{n+1})=1,
\end{align}
where $i$ labels the nodes directly above $\lc{C}_n$. Denote the amplitude of reaching $\lc{C}_n$ by $A(\lc{C}_n)$. Then $A$ defines a complex function (a ``pre-measure'') on the cylinder sets whose extension to a complex measure on $\lc{\mathcal{R}}$---\textit{if it exists}---is the desired complex measure from which the decoherence functional can be obtained. The following criteria for the existence of an extension to a complex measure on $\lc{\mathcal{R}}$ were given in \cite{Surya:2020cfm}. Let $\zeta$ be the following real function on the ground-set of labeled poscau, \begin{align}\label{eq_04071901}\begin{split}
 &\zeta(\lc{C}_n):=\sum_{i} |A(\lc{C}_{n}\rightarrow \lc{C}^i_{n+1})|-|\sum_{i} A(\lc{C}_{n}\rightarrow \lc{C}^i_{n+1})|\\
 & \hspace{11mm}=\sum_{i} |A(\lc{C}_{n}\rightarrow \lc{C}^i_{n+1})|-1,\\\end{split}\end{align}
where the equality follows from \eqref{eq_sumrule_sec5}.
For each $n>0$, define the maximum and minimum of $\zeta$ over level $n$ in labeled poscau (where level 1 contains the single-element causet, level 2 contains the 2-element causets etc.),
\begin{align}\label{eq_04071901_2}\begin{split}
&\zeta_n^{max}:=\max_{\text{level } n} \zeta({\lc{C}_n}),\\
&\zeta_n^{min}:=\min_{\text{level } n} \zeta({\lc{C}_n}).\end{split}\end{align}
Then, an extension to a complex measure on $\lc{\mathcal{R}}$ exists if $\sum_{n=1}^{\infty} \zeta^{max}_n<\infty$. An extension to a complex measure on $\lc{\mathcal{R}}$ does not exist if $\sum_{n=1}^{\infty} \zeta^{min}_n=\infty$.

We now apply this technology to the discussion of Complex Transitive Percolation, a family of cyclic models defined by \begin{align}\label{CTP} A(\lc{C}_n)=p^Lq^{\binom{n}{2}-R},\end{align} where $p\in\mathbb{C}$, $q=1-p$, and $L$ and $R$ are the number of links and relations in $\lc{C}_n$,  respectively. In these models, an extension to a measure on $\lc{\mathcal{R}}$ exists if and only if $p\in[0,1]$ (\textit{i.e.} when the amplitude $A$ reduces to a real probability) \cite{Dowker:2010qh,Surya:2020cfm}. But, with our new understanding of $\mathcal{R}_b$ as the algebra of covariant events in this model, we can ask whether there exists an extension to the strictly smaller $\lc{\mathcal{R}}_b$ when $p\not\in[0,1]$. (\textit{Note that this question is equivalent to asking whether there are measures on $\lc{\mathcal{R}}_b$ that  do not extend to a measure on $\lc{\mathcal{R}}$. When the measures are real probability measures, the answer is no: every measure on $\lc{\mathcal{R}}_b$ extens to a measure on $\lc{\mathcal{R}}$. However, here we are concerned with complex measures, and it is unknown to the authors whether an analogous theorem holds in this case.}) To this effect, we can apply equations (\ref{eq_sumrule_sec5}-\ref{eq_04071901_2}) and the above criteria for extension by replacing $\lc{\mathcal{R}}$ with $\lc{\mathcal{R}}_b$ and the $\lc{C}_n$ with the nodes of reduced poscau. This generalisation is possible since the results of \cite{Surya:2020cfm} rely only on the fact that labeled poscau is a finite-valency tree with no maximal elements, see \cite{Zalel:2021gkv} for further discussion. The transition amplitudes on reduced poscau are fixed by requiring that \eqref{CTP} holds for principal causets and by imposing \eqref{eq_sumrule_sec5}. We solved numerically for $\zeta_n^{min}$ as a function of $p$ for $n=2,3,4$. Our results (Fig.\ref{plots_b}) suggest that for any $p$, there is a level $m$ above which the function $\zeta$ takes its minimum on the nodes that  contain principal causets, \textit{i.e.} $\zeta_n^{min}=|p|+|1-p|-1$ for all $n>m$. If this is borne out then $\sum_{n=1}^{\infty} \zeta^{min}_n=\infty$ when $p\not\in[0,1]$ and no extension to $\lc{\mathcal{R}}_b$ exists, meaning that Complex Transitive Percolation does not give rise to a well-defined quantum dynamics in this framework. Whether other classical cyclic dynamics can give rise to quantum dynamics via this formalism is an open question. 

\begin{figure}[h]
    \centering
\begin{subfigure}{.1\textwidth}
    \includegraphics[width=0.5\textwidth]{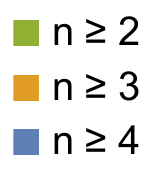}
    \end{subfigure}\hfill%
    \begin{subfigure}{.4\textwidth}
    \includegraphics[width=.8\textwidth]{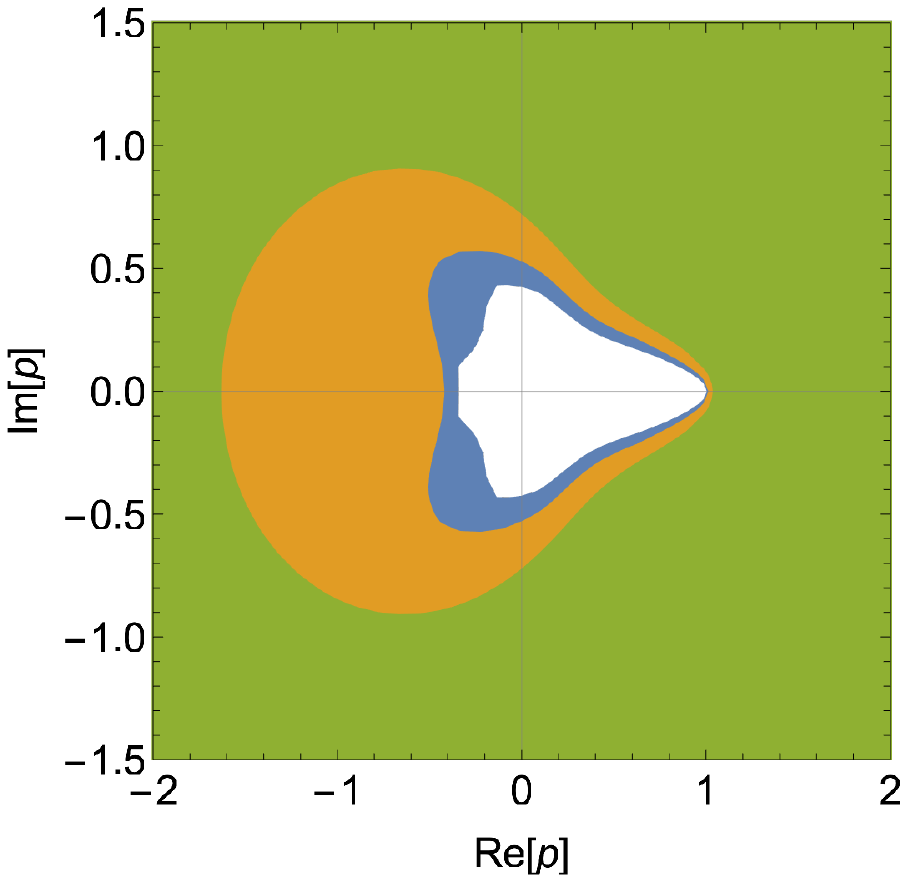}
    \end{subfigure}\hfill%
    \caption[Complex Transitive Percolation with breaks]{
    The coloured regions indicate the values of $p\in\mathbb{C}$ for which $\zeta_n^{min}=|p|+|1-p|-1$ on reduced poscau for $n=2,3,4$ in Complex Transitive Percolation models.}%
    \label{plots_b}%
\end{figure}

\section{Observables from principal stem-sets}\label{sec2}

In section \ref{chap_6_o_subsec} we identified $\mathcal{R}_b$ as a complete set of observables in cyclic dynamics, but it is unclear what is the physical meaning of each observable in this set. In this section, we identify a second algebra that  forms a complete set of observables in cyclic dynamics and in which each observable has a clear physical interpretation. 

We begin with terminology. We say that an order $C$ is a \textit{principal order} if its representative is a principal causet. We say that $stem(C)$ is a \textit{principal stem-set} if $C$ is a principal order. We define $\widehat{\mathcal{S}}\subset\mathcal{S}$ to be the set of principal stem-sets and write $\mathcal{R}(\widehat{\mathcal{S}})$ to denote the $\sigma$-algebra that  they generate. We can now state our theorem:

\begin{theorem}\label{theorem3} Given $\mathcal{E}\in\mathcal{R}$, there exists some $\mathcal{E}'\in \mathcal{R}(\widehat{\mathcal{S}})$ such that $\mathcal{E}\triangle\mathcal{E}'\subset\mathcal{B}_{\infty}^c$.\end{theorem}

In perfect analogy to the result of \cite{Brightwell:2002vw}, the upshot of lemma \ref{theorem3} is that $\mathcal{R}(\widehat{\mathcal{S}})$ exhausts the set of observables in any cyclic dynamics, since the measure of any $\mathcal{E}\in\mathcal{R}$ is fixed by the measure of some $\mathcal{E}'\in\mathcal{R}(\widehat{\mathcal{S}})$ via $\mu(\mathcal{E})=\mu(\mathcal{E}')$.

Importantly, each event in $\mathcal{R}(\widehat{\mathcal{S}})$ is equivalent to a logical combination of statements about which principal orders are contained as stems in the growing causet, giving the observables in cyclic dynamics a clear physical interpretation. We will make the notion of “an order contained as a stem” precise after lemma \ref{lemma_0706}.

We now prove theorem \ref{theorem3} through a series of lemmas.

\begin{lemma}\label{lemma1506} Let $X$ denote a set of points and let $\sim$ denote an equivalence relation on $X$. Let $(X,R)$ denote a standard Borel space. Let $(X,T)$ denote the Borel space derived from $(X,R)$ via the ``covariance'' property:
\begin{align}\label{eq1506}T:=\{E\in R | a\in E \text{ and } b\sim a \implies b\in E\}. \end{align} Let $F\subset T$ denote a countable family of sets that  separates $X$ up to equivalence under $\sim$, \textit{i.e.} for any $a\not\sim b\in X$ there exists a set in $F$ that  contains $a$ but not $b$. Then $F$ generates $T$. \end{lemma}

\begin{proof} Let $X'$ denote the set of equivalence classes of elements of $X$ under $\sim$ and let $\phi:X\rightarrow X'$ denote the projection that  maps each $a\in X$ onto its equivalence class. Given $(X,R)$, the quotient Borel space induced by $\phi$ is denoted by $(X',R')$, where $R'$ is the set of all $A'\subset X'$ such that $\phi^{-1}(A')\in R$. 

The quotient Borel space $(X',R')$ is analytic whenever it is countably separated (second theorem on p.74 in \cite{Mackey:1976}). In an analytic Borel space, any countable separating family is a generating family (corollary on p.73 in \cite{Mackey:1976}). Therefore, any countable separating family in $(X',R')$ is a generating family.

Since $T$ is derived from $R$ via the ``covariance'' property, $\phi$ induces a bijection between $T$ and $R'$. Under this bijection, $F\subset T$ is mapped onto a countable separating family $F'$ in $R'$. Since $F'$ generates $R'$, and since bijections preserve countable set operations, $F$ generates $T$. \end{proof}

\begin{lemma}\label{lemma_0706} For any $\lc{D}\not\cong \lc{E}\in \mathcal{B}_{\infty}$ there exists a principal order $C_n$ such that $\lc{D}\in stem(C_n)$ and $\lc{E}\not\in stem(C_n)$. \end{lemma}

We will need the following terminology. Given a finite order $C$ and a (finite or infinite) causet $\lc{D}$, we say that $C$ is a stem in $\lc{D}$ if there exists a stem in $\lc{D}$ that  is order-isomorphic to a representative of $C$. If $C$ is a stem in $\lc{D}$ and $\lc{D}$ is a representative of $D$, then we say that $C$ is a stem in $D$. The cardinality $|C|$ of an order $C$ is defined to be equal to the cardinality of a representative of it. We use the term $n$-stem to mean a stem of cardinality $n$. We say than an order contains a break if its representative contains a break. The covering causet of $\lc{C}$ is the principal causet formed from $\lc{C}$ by placing an element above all $x\in\lc{C}$. The covering order $\widehat{C}$ of $C$ is the order whose representative is a covering causet of some representative of $C$.

\begin{proof}
Let $\lc{D},\lc{E}\in\mathcal{B}_{\infty}$. Let $D_i$ denote the unlabeled past of the $i^{th}$ break in $\lc{D}$, so $D_i$ is an order that is a stem in $\lc{D}$. Note that $D_i$ may or may not be principal. Define $E_i$ similarly for $\lc{E}$. Note that $D_i$ and $E_i$ each contain $i-1$ breaks and therefore either $D_i=E_i$ or $D_i\not= E_k \forall k\in\mathbb{N}$.

Let $\lc{D}\not\cong\lc{E}$. Then there exists an $n\in\mathbb{N}$ such that $D_i\not= E_i$ for all $i\geq n$.

Without loss of generality, let $|D_n|\geq|E_n|$. We will show that $\lc{E}\notin stem(\widehat{D_n})$. Suppose for contradiction that $\widehat{D_n}$ is a stem in $\lc{E}$. Then $\widehat{E_n}$ is the only $|\widehat{E_n}|$-stem in $\widehat{D_n} \implies \widehat{E_n}$ is the only $|\widehat{E_n}|$-stem in $\lc{D}$ (since $\widehat{D_n}$ is the only $|\widehat{D_n}|$-stem in $\lc{D}$, and ``not a stem in any stem is not a stem'')$\implies E_n$ is an unlabeled past of a break in $\lc{D}$$\implies$ there exists some $i\in\mathbb{N}$ such that $D_i=E_n$, which is a contradiction.
\end{proof}

\begin{corollary}\label{cor1506} The countable family of sets, \begin{align}\label{1807} \widehat{\mathcal{S}}\cap \mathcal{B}_{\infty}:=\{stem(C)\cap \mathcal{B}_{\infty}|stem(C)\in \widehat{\mathcal{S}}\},\end{align} separates $\mathcal{B}_{\infty}$ up to equivalence under order-isomorphisms, \textit{i.e.} for any $\lc{D}\not\cong \lc{E}\in \mathcal{B}_{\infty}$ there exists a set in $\widehat{\mathcal{S}}\cap \mathcal{B}_{\infty}$ that  contains $\lc{D}$ but not $\lc{E}$. \end{corollary}

\begin{lemma} $\mathcal{R}(\widehat{\mathcal{S}})\cap \mathcal{B}_{\infty}={\mathcal{R}}\cap\mathcal{B}_{\infty}$. \end{lemma}
\begin{proof}
We will show that both the LHS and the RHS are generated by the family $\widehat{\mathcal{S}}\cap \mathcal{B}_{\infty}$, and the result follows.

We begin with the RHS. Note that $(\lc{\Omega},\lc{\mathcal{R}})$ is a Polish space (lemma 6 in \cite{Brightwell:2002vw}), and therefore its subspace $(\mathcal{B}_{\infty}, \lc{\mathcal{R}}\cap\mathcal{B}_{\infty})$ is a standard Borel space (definition 1 on p.71 in \cite{Mackey:1976}). The Borel space derived from $(\mathcal{B}_{\infty}, \lc{\mathcal{R}}\cap\mathcal{B}_{\infty})$ via the ``covariance'' property \eqref{eq1506} with respect to equivalence under order-isomorphism is equal to $(\mathcal{B}_{\infty}, {\mathcal{R}}\cap\mathcal{B}_{\infty})$ and contains the countable family of sets $\widehat{\mathcal{S}}\cap \mathcal{B}_{\infty}$. Together, corollary \ref{cor1506} and lemma \ref{lemma1506} imply that the family $\widehat{\mathcal{S}}\cap \mathcal{B}_{\infty}$ generates $(\mathcal{B}_{\infty}, {\mathcal{R}}\cap\mathcal{B}_{\infty})$.

For the LHS, note that if a Borel space $(X,S)$ is generated by a family ${F}$ then its Borel subspace $(Q,S\cap Q)$ is generated by the family ${F}\cap Q$.
\end{proof}

\begin{proof}[Proof of theorem \ref{theorem3}]
Consider some $\mathcal{E}\in \mathcal{R}$. Then $\mathcal{E}\cap\mathcal{B}_{\infty} \in {\mathcal{R}}\cap\mathcal{B}_{\infty}\implies \mathcal{E}\cap\mathcal{B}_{\infty}\in\mathcal{R}(\widehat{\mathcal{S}})\cap \mathcal{B}_{\infty}\implies \exists \mathcal{E}'\in \mathcal{R}(\widehat{\mathcal{S}})$ such that $\mathcal{E}'\cap\mathcal{B}_{\infty}=\mathcal{E}\cap\mathcal{B}_{\infty}$, \textit{i.e.} $\mathcal{E}\triangle \mathcal{E}' \subset \mathcal{B}^c_{\infty}$. \end{proof}

\section{Posts}\label{sec:posts}

In the special case when the epochs are separated by posts, the dynamics satisfy $\mu(\mathcal{P}_{\infty})=1$, where $\mathcal{P}_{\infty}\in\mathcal{R}$ is the event that the causet spacetime contains infinitely many posts, and the $\sigma$-algebra,
\begin{align} \mathcal{R}_p:=\{\mathcal{E} \in\mathcal{R}|\mathcal{P}_{\infty}^c\subseteq \mathcal{E}\text{ or }\mathcal{P}_{\infty}^c \subseteq \mathcal{E}^c\}\end{align}
exhausts the covariant events, by analogy to the discussion of \eqref{roguealgebra} and \eqref{18075}. The strictly stronger analogue of theorem \ref{theorem_sec_6_1} is,
\begin{theorem}\label{theorem_sec_6_2}$\mathcal{R}_p=\lc{\mathcal{R}}_p\cap{\mathcal{R}}$. \end{theorem}
Here, $\lc{\mathcal{R}}_p$ denotes the $\sigma$-algebra generated by the cylinder sets associated with those principal causets $\lc{C}_n$ whose restriction $\lc{C}_n|_{[0,n-2]}$ is itself a principal causet. We call such a causet $\lc{C}_n$, and the cylinder set, order and stem-set associated with it, \textit{doubly principal}. The analogue of theorem \ref{theorem3} is, 
\begin{theorem}\label{theorem3:posts} Given $\mathcal{E}\in\mathcal{R}$, there exists some $\mathcal{E}'$ in the $\sigma$-algebra generated by the doubly principal stem-sets such that $\mathcal{E}\triangle\mathcal{E}'\subset\mathcal{P}_{\infty}^c$.\end{theorem}

Let us define a \textit{principal break} to be a break whose past is a principal stem. Then there is a post if and only if there is a principal break: the post is the maximal element of the past of the principal break. The proof of theorems \ref{theorem_sec_6_2} and \ref{theorem3:posts} can be obtained from the proofs of theorems \ref{theorem_sec_6_1} and \ref{theorem3} respectively, mutatis mutandis \textit{i.e.}  by replacing “break” with “principal break” and “principal stem” with “doubly principal stem”. So, for example, a segment is now defined to be the portion of the causet between two principal breaks and is always a principal causet.

Measures on ${\lc{\mathcal{R}}_p}$ correspond to walks up \textit{doubly-reduced poscau} whose description is obtained from the description of reduced poscau (see section \ref{subsec:reducedposc}) by replacing ``principal'' with ``doubly principal'' (so only doubly principal causets are contained in nodes of their own). Thus, cyclic models in which the epochs are separated by posts can be conceived of as random walks on this novel tree. Complex Transitive Percolation is such a cyclic model, and so we can extend the discussion in section \ref{sec:CTP} by asking whether there exists an extension of the measure to the strictly smaller $\sigma$-algebra $\lc{\mathcal{R}}_p\subset\lc{\mathcal{R}}_b$. We do so by modifying equations (\ref{eq_sumrule_sec5}-\ref{eq_04071901_2}) and the criteria for extension by replacing $\lc{\mathcal{R}}$ with $\lc{\mathcal{R}}_p$ and the $\lc{C}_n$ with the nodes of doubly-reduced poscau, but their application to Complex Transitive Percolation is inconclusive. On the one hand, we cannot prove that an extension exists since $\zeta$ is equal to $|p|+|1-p|-1$ on every doubly principal causet so $\zeta_n^{max}\geq |p|+|1-p|-1> 0$ for all $n$ when $p\not\in[0,1]$. On the other hand, we can cannot rule out an extension, since every level $n>1$ in doubly-reduced poscau contains nodes with valency equal to 1 and on these nodes $\zeta$ vanishes and therefore $\sum_n\zeta_n^{min}=0$.

We now amend the extension criteria to provide further scrutiny in the special case where there are nodes with valency equal to 1. Let $\mathscr{T}$ denote a finite-valency directed tree that  contains no maximal elements and let $\mathscr{T}_n \subset \mathscr{T}$ denote the set of nodes at level $n$. Let $D_n$ denote a node at level $n$.

Define, $S_n:= \sum_{D_n\in\mathscr{T}_n} |A(D_n)|$, and note that an extension exists if and only if $\sup_n S_n<\infty$ \cite{Surya:2020cfm}. 

Let $\overline{\mathscr{T}_n}\subset\mathscr{T}_n$ be the set of level $n$ nodes that  have valency greater than 1, and define,
\begin{align}\label{def230222}\begin{split}&\overline{\zeta_n^{min}}:=\min_{D_n\in\overline{\mathscr{T}_n}}\zeta(D_n),\\&\overline{S_{n}}:= \sum_{D_{n}\in\overline{\mathscr{T}_{n}}} |A(D_{n})|,\\
&S_{n}^{v=1}:=S_n-\overline{S_{n}}.
\end{split}\end{align}
Then one can use definitions \eqref{def230222} to generalise the proof of claim 3.2 in \cite{Surya:2020cfm} and thus show that,
\begin{align}\label{eq_lemma241_sec5_chap_6_3}
\begin{split}
S_n\geq & \prod_{r=1}^{n-1}(1+\overline{\zeta_{r}^{min}})-\sum_{r=1}^{n-1} \bigg[S_{n-r}^{v=1}\  \overline{\zeta_{n-r}^{min}}\  \prod_{i=1}^{r-1}(1+\overline{\zeta_{n-r+i}^{min}})\bigg].\\
\end{split}
\end{align}
Thus, if the right hand side of \eqref{eq_lemma241_sec5_chap_6_3} diverges with $n$ then no extension exists.

Applying this improved criterion to Complex Transitive Percolation on doubly-reduced poscau, our numerical solutions for $\overline{\zeta_n^{min}}$ as a function of $p$ for $n=2,3,4$ (Fig.\ref{plots_p}) suggest that for any $p\in\mathbb{C}$, there exists a level $m$ above which the function $\zeta$ restricted to the valency $>1$ nodes takes its minimum value on the doubly principal causets, \textit{i.e.} $\overline{\zeta_n^{min}}=|p|+|1-p|-1$ for all $n>m$. If this is borne out then, when $p\not\in[0,1]$, the first term on the right hand side of \eqref{eq_lemma241_sec5_chap_6_3} diverges as $n\rightarrow\infty$. Whether $\sup_n S_n=\infty$ depends on the behaviour of $S_{n-r}^{v=1}$ in the second term. We note that, since $S_{n-r}^{v=1}$ is the sum over absolute values of amplitudes of reaching a doubly principal causet by stage $n-r-1$, in future it could be computed using techniques similar to those used to obtain the probability of a post \cite{Bombelli:2008kr}.

\begin{figure}[h]
    \centering
\begin{subfigure}{.1\textwidth}
    \includegraphics[width=0.5\textwidth]{legend.pdf}
    \end{subfigure}\hfill%
    \begin{subfigure}{.4\textwidth}
    \includegraphics[width=.8\textwidth]{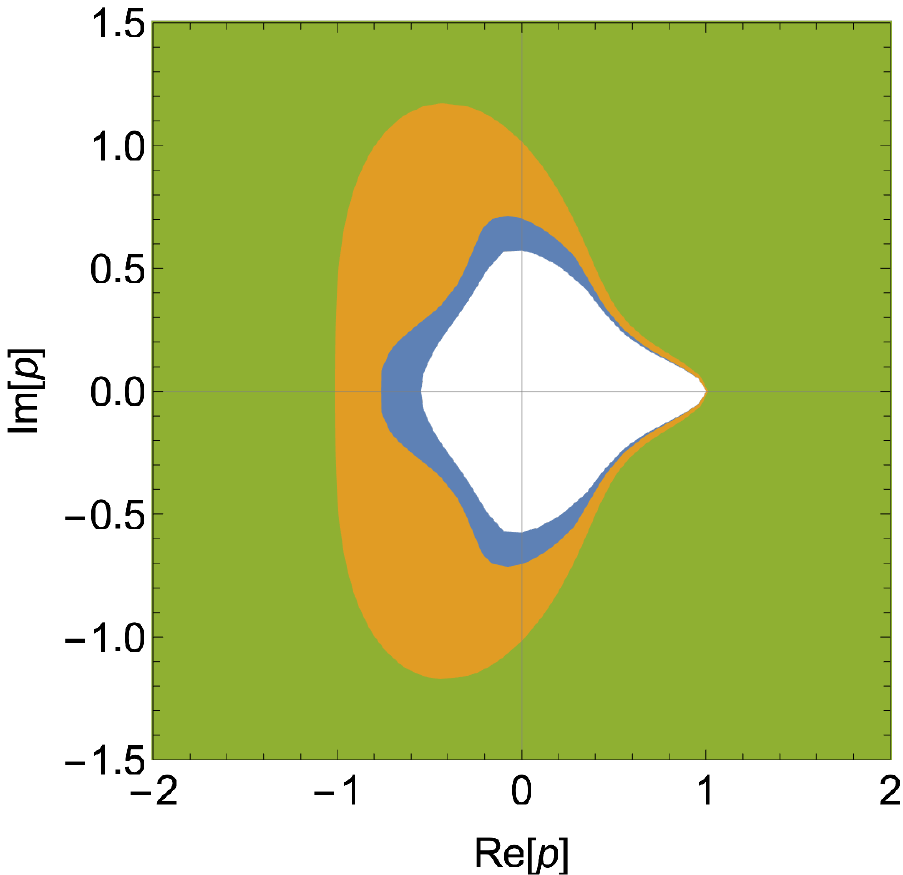}
   \caption{Doubly-reduced poscau.}\label{plot2}
    \end{subfigure}\hfill%
    \caption[Complex Transitive Percolation with posts]{
    The coloured regions indicate the values of $p\in\mathbb{C}$ for which  $\overline{\zeta_n^{min}}=|p|+|1-p|-1$ on doubly-reduced poscau for $n=2,3,4$ in Complex Transitive Percolation models.}%
    \label{plots_p}%
\end{figure}

\section{Conclusion}
In this work, we considered how the set of covariant observables can be distilled to a smaller exhaustive set of observables under the family of cyclic dynamics. This interplay between the kinematic constraint of covariance and the dynamic restriction to on-shell configurations (cyclic causets) gives rise to interrelated systems of sets with rich mathematical structure. In particular, we identified both $\mathcal{R}_b$ and $\mathcal{R}(\widehat{\mathcal{S}})$ as exhaustive observable $\sigma$-algebras, and it is natural to ask how the two are related. First, note that both are sub-algebras of $\mathcal{R}(\mathcal{S})$, since neither separates the set of rogues. Additionally, definition \eqref{18075} of $\mathcal{R}_b$ and theorem \ref{theorem3} combine to imply that: \textit{for any $\mathcal{E}\in\mathcal{R}_b$, there exists some $\mathcal{E}'\in \mathcal{R}(\widehat{\mathcal{S}})$ such that $\mathcal{E}\triangle\mathcal{E}'\subset\mathcal{B}_{\infty}^c$, and vice versa, for any $\mathcal{E}\in\mathcal{R}(\widehat{\mathcal{S}})$, there exists some $\mathcal{E}'\in\mathcal{R}_b $ such that $\mathcal{E}\triangle\mathcal{E}'\subset\mathcal{B}_{\infty}^c$}. Therefore, under cyclic dynamics the two algebras of observables are equal up to sets of measure zero and can be considered equivalent. However, at the level of the kinematics the two $\sigma$-algebras are very different since the only events that  they share are the unit and the empty set, as we now prove.


\begin{lemma}\label{18077}  Let $\mathcal{E}\in \mathcal{R}({\mathcal{S}})$. If $\mathcal{E}$ is contained in both $\mathcal{R}_b$ and $\mathcal{R}(\widehat{\mathcal{S}})$ then $\mathcal{E}$ is either the empty set or the unit element $\lc{\Omega}$.\end{lemma}

\begin{proof}
Consider some causet $\lc{C}\in\mathcal{B}_{\infty}$ and let $C_i$ denote the unlabeled past of the $i^{th}$ break in $\lc{C}$, so $C_i$ is an order which is a stem in $\lc{C}$. 
Define $\lc{C}'$ to be the causet which is some labeling of the disjoint union of the $C_i$'s. The event $\bigcap_i stem(\widehat{C_i})$ is the smallest event in $\mathcal{R}(\widehat{\mathcal{S}})$ that contains $\lc{C}$, in the sense that if $\lc{C}\in\mathcal{E}\in\mathcal{R}(\widehat{\mathcal{S}})$ then $\bigcap_i stem(\widehat{C_i})\subset\mathcal{E}$. The event $\bigcap_i stem(\widehat{C_i})$ is also the smallest event in $\mathcal{R}(\widehat{\mathcal{S}})$ that contains $\lc{C}'$. Therefore any $\mathcal{E}\in\mathcal{R}(\widehat{\mathcal{S}})$ contains either both or neither of $\lc{C}$ and $\lc{C}'$.

Suppose $\mathcal{E}\in \mathcal{R}_b$. By definition of $\mathcal{R}_b$, either $\mathcal{E}\subseteq\mathcal{B}_{\infty}$ or $\mathcal{B}_{\infty}^c\subseteq\mathcal{E}$. We show that in both cases, $\mathcal{E}\not\in\mathcal{R}(\widehat{\mathcal{S}})$, which completes the proof.

\underline{Case(i): $\mathcal{E}\subseteq\mathcal{B}_{\infty}$ and $\mathcal{E}\not=\emptyset$.} Note that $\mathcal{E}$ contains some $\lc{C}\in\mathcal{B}_{\infty}$. Therefore, if $\mathcal{E}\in \mathcal{R}(\widehat{\mathcal{S}})$ then $\mathcal{E}$ contains $\lc{C}'\not\in\mathcal{B}_{\infty}$. Contradiction.

\underline{Case(ii): $\mathcal{B}_{\infty}^c\subseteq\mathcal{E}$ and $\mathcal{E}\not=\lc{\Omega}$.} Note that there exists some causet $\lc{C}\in\mathcal{B}_{\infty}$ that  is not contained in $\mathcal{E}$ and thefore $\lc{C}'\not\in\mathcal{E}$. Contradiction. \end{proof}

In particular, lemma \ref{18077} implies that $\mathcal{B}_{\infty}\not\in\mathcal{R}(\widehat{\mathcal{S}})$. Therefore one cannot tell from the measure on $\mathcal{R}(\widehat{\mathcal{S}})$ whether the dynamics is cyclic, however if one knows the dynamics is cyclic then $\mathcal{R}(\widehat{\mathcal{S}})$ is an exhaustive set of observables. On the other hand, since $\mathcal{B}\in\mathcal{R}_b$, knowing the measure on $\mathcal{R}_b$ is sufficient to determine whether the dynamics is cyclic.

Finally, our motivation for studying cyclic dynamics has been the key role that they play in the the causal set comological paradigm that  aims to explain the emergence of a flat, homogeneous and isotropic cosmos directly from the quantum gravity era \cite{Sorkin:1998hi,Martin:2000js,Dowker:2017zqj}. But what form do these models take and do they occupy a significant volume in theory space? The Transitive Percolation family of cyclic CSG model \cite{Alon:1994, Rideout:1999ub} has served as a starting point for searches of cyclic models and some progress has been made in this direction: a class of cyclic dynamics in which epochs are separated by posts has been identified in \cite{Ash:2005za} and a conjecture of another class of such models has been put forward in \cite{Brightwell:2016}.  
But these formal results are yet to be fully understood and implemented (\textit{e.g.}, via computer simulations) and the characterisation of cyclic dynamics remains an important open question. 

\textbf{Acknowledgements:} This research was supported in part by STFC grant ST/T000791/1. It was also
supported in part by Perimeter Institute for Theoretical Physics. Research at Perimeter
Institute is supported by the Government of Canada through Industry Canada and by
the Province of Ontario through the Ministry of Economic Development and Innovation.
\bibliography{observ_bib}{}

\end{document}